\documentclass[12pt]{article}

\usepackage{amsmath,amssymb,amsthm,cite,mathtools,slashed}
\usepackage[margin=2.5cm]{geometry}
\usepackage{xcolor}

\usepackage[full]{textcomp}
\usepackage[osf]{newtxtext} 
\usepackage[bigdelims,nosymbolsc,cmintegrals,vvarbb]{newtxmath}
\usepackage[cal=euler,frak=euler]{mathalfa} 
\usepackage{bm} 

\title{The Kirillov picture for the Wigner particle}

\author{
J. M. Gracia-Bond\'ia$^1$,
F. Lizzi$^{2,3,4}$,
J. C. V\'arilly$^5$
and P. Vitale$^{2,3}$
\\ \\[12pt]
{\footnotesize $^1$ Departamento de F\'isica Te\'orica, Universidad de
Zaragoza, Zaragoza 50009, Spain}\\[3pt]
{\footnotesize $^2$ Dipartimento di Fisica ``E.~Pancini'',
Universit\`a di Napoli Federico II, Italy}\\[3pt]
{\footnotesize $^3$ INFN Sezione di Napoli, Italy}\\[3pt]
{\footnotesize $^4$ Departament de F\'isica Qu\`antica i Astrof\'isica
\& Institut de Ci\`encies del Cosmos, Universitat de Barcelona,
Spain}\\[3pt]
{\footnotesize $^5$ Escuela de Matem\'atica,
Universidad de Costa Rica, San Jos\'e 11501, Costa Rica}\\[3pt]
}

\date{May 3, 2018}

\DeclareMathOperator{\Ad}{Ad}       
\DeclareMathOperator{\ad}{ad}       
\DeclareMathOperator{\Coad}{Coad}   
\DeclareMathOperator{\coad}{coad}   
\DeclareMathOperator{\sech}{sech}   

\newcommand{\al}{\alpha}            
\newcommand{\bt}{\beta}             
\newcommand{\dl}{\delta}            
\newcommand{\eps}{\varepsilon}      
\newcommand{\ga}{\gamma}            
\newcommand{\ka}{\kappa}            
\newcommand{\La}{\Lambda}           
\newcommand{\la}{\lambda}           
\newcommand{\om}{\omega}            
\newcommand{\sg}{\sigma}            
\renewcommand{\th}{\theta}          
\newcommand{\ze}{\zeta}             

\newcommand{\bR}{\mathbb{R}}        
\newcommand{\bS}{\mathbb{S}}        

\newcommand{\sL}{\mathcal{L}}       
\newcommand{\sO}{\mathcal{O}}       
\newcommand{\sP}{\mathcal{P}}       

\newcommand{\pl}{\mathfrak{p}}      

\newcommand{\del}{\partial}         
\newcommand{\downto}{\downarrow}    
\newcommand{\hel}{\lambda}          
\newcommand{\less}{\setminus}       
\newcommand{\lt}{\triangleright}    
\newcommand{\unl}{\underline}       
\newcommand{\up}{{\mathord{\uparrow}}} 
\newcommand{\upto}{\uparrow}        
\newcommand{\w}{\wedge}             
\newcommand{\x}{\times}             
\renewcommand{\.}{\cdot}            

\newcommand{\half}{{\mathchoice{\thalf}{\thalf}{\shalf}{\shalf}}}
\newcommand{\shalf}{{\scriptstyle\frac{1}{2}}} 
\newcommand{\thalf}{\tfrac{1}{2}}   

\bmdefine{\aaa}{a}          
\bmdefine{\bbb}{b}          
\bmdefine{\eee}{e}          
\bmdefine{\ggg}{g}          
\bmdefine{\GG}{G}           
\bmdefine{\ii}{i}           
\bmdefine{\jj}{j}           
\bmdefine{\JJ}{J}           
\bmdefine{\kk}{k}           
\bmdefine{\KK}{K}           
\bmdefine{\lll}{l}          
\bmdefine{\LL}{L}           
\bmdefine{\mm}{m}           
\bmdefine{\nn}{n}           
\bmdefine{\PP}{P}           
\bmdefine{\pp}{p}           
\bmdefine{\QQ}{Q}           
\bmdefine{\qq}{q}           
\bmdefine{\RR}{R}           
\bmdefine{\rr}{r}           
\bmdefine{\sss}{s}          
\bmdefine{\SSS}{S}          
\bmdefine{\ttt}{t}          
\bmdefine{\TT}{T}           
\bmdefine{\uu}{u}           
\bmdefine{\UU}{U}           
\bmdefine{\vv}{v}           
\bmdefine{\VV}{V}           
\bmdefine{\ww}{w}           
\bmdefine{\WW}{W}           
\bmdefine{\xx}{x}           
\bmdefine{\XX}{X}           
\bmdefine{\yy}{y}           
\bmdefine{\zz}{z}           
\bmdefine{\zero}{0}         
\bmdefine{\omom}{\omega}    
\bmdefine{\ssg}{\sigma}     

\newcommand{\hideqed}{\renewcommand{\qed}{}} 
\newcommand{\set}[1]{\{\,#1\,\}}     
\newcommand{\word}[1]{\quad\text{#1}\quad} 
\def\wick:#1:{\,\mathopen:#1\mathclose:\,} 

\newcommand{\pd}[2]{\frac{\partial#1}{\partial#2}} 
\def\duo<#1,#2>{\langle#1,#2\rangle} 
\def\lLa^#1_#2{\Lambda^{#1}{}_{#2}}  
\def\ltLa_#1^#2{\Lambda_{#1}{}^{#2}} 

\def\epsi^#1_#2{\eps^{#1}{}_{\!#2}} 
\def\epsii_#1^#2{\eps_{#1}{}^{\!#2}} 

\theoremstyle{plain}
\newtheorem{thm}{Theorem}            
\newtheorem{prop}[thm]{Proposition}  
\newtheorem{lema}[thm]{Lemma}        
\newtheorem{corl}[thm]{Corollary}    

\theoremstyle{definition}

\numberwithin{equation}{section}

\makeatletter
\renewcommand{\section}{\@startsection{section}{1}{\z@}%
                       {-3.5ex \@plus -1ex \@minus -.2ex}%
                       {2.3ex \@plus.2ex}%
                       {\normalfont\large\bfseries}}
\renewcommand{\subsection}{\@startsection{subsection}{2}{\z@}%
                        {-3.25ex \@plus -1ex \@minus -.2ex}%
                        {1.5ex \@plus .2ex}%
                        {\normalfont\normalsize\bfseries}}
\renewcommand{\subsubsection}{\@startsection{subsubsection}{3}{\z@}%
                        {-3.25ex \@plus -1ex \@minus -.2ex}%
                        {1.5ex \@plus .2ex}%
                        {\normalfont\small\bfseries}}
\makeatother

\begin{document}

\maketitle

\begin{abstract}
We discuss the Kirillov method for massless Wigner particles, usually
(mis)named ``continuous spin'' or ``infinite spin'' particles. These
appear in Wigner's classification of the unitary representations of
the Poincar\'e group, labelled by elements of the enveloping algebra
of the Poincar\'e Lie algebra. Now, the coadjoint orbit procedure
introduced by Kirillov is a prelude to quantization. Here we exhibit
for those particles the classical Casimir functions on phase space, in
parallel to quantum representation theory. A good set of position
coordinates are identified on the coadjoint orbits of the Wigner
particles; the stabilizer subgroups and the symplectic structures of
these orbits are also described.
\end{abstract}

\section{Introduction}
\label{sec:introibo}

The Wigner unitary representations (unirreps) of the Poincar\'e
group~\cite{Wigner39}, describing relativistic elementary particles,
count, in our opinion, among the most important mathematical objects
in the whole of physics. The Kirillov coadjoint orbit
picture~\cite{Kirillov04}, on the other hand, has been known (for over
a half century now) to link symplectic geometry with harmonic
analysis. It is therefore surprising that relatively little work has
been done so far to relate the Wigner unirreps with the phase space
orbits (homogeneous symplectic manifolds endowed with a canonical
Liouville measure) for the Poincar\'e group. The surprise only grows
when one realizes that their correspondence is one-to-one,
particularly for maximal dimension orbits, like the ones considered
here. A partial exception was the paper by J.~F.~Cari\~nena and two of
us~\cite{Ganymede}, devoted to phase spaces corresponding to
\textit{massive} particles. There, moreover, physical quantum averages
were computed by means of phase space integrals, in a Wigner--Moyal
approach.

It is our view that classical elementary relativistic phase spaces are
objects as intrinsic as -- and perhaps more readily understandable
than -- the corresponding Wigner unirreps. Observables defined on
elementary classical systems are useful proxies for operator
quantities, since relativistic kinematics is the same for classical
and quantum objects. In connection with the quantum field theory
procedures, it should be noted that quantization of the coadjoint
orbit picture runs parallel to induced representation
theory~\cite{Li93}, and allows to recover many of its results.

In this vein, we examine in this paper the phase space counterparts to
Wigner's so-called ``continuous spin'' or ``infinite spin''
representations. These are misnomers (there is no such thing as
continuous or infinite spin, although ``unbounded'' passes muster), so
we shall call them simply \textit{Wigner particles}, or WP for short.
Till recently the latter had received scant attention, being curtly
dismissed in textbooks. However, the path-breaking series of papers on
the WP properties by Schuster and Toro
\cite{SToro13a,SToro13b,SToro13c,SToro15} has awakened a lot of
interest~\cite{RempelF16,BekaertS17,Najafizadeh17}; we retain chiefly
\cite{Rehren17}, which introduces a quantum stress-energy-momentum
tensor for the~WP.

To the best of our knowledge, this article is the first study of the
coadjoint orbits for \textit{the} Wigner particle. We work out in
detail the Poisson bracket structure for these \textit{lightlike}
systems. A crucial task is to find, and to establish the properties
of, good position functions on the orbits.

The paper is organized as follows. Section~\ref{sec:maskirovka}
recalls the basics of the Kirillov procedure, as applied to the
Poincar\'e group $\sP$. We find the classical Casimir functions on
phase space, in striking parallelism with quantum representation
theory. Their values index the orbits corresponding to such
representations, on which a convenient set of coordinates for the
description of WPs is found. Subsection~\ref{ssc:gira-el-mundo}
dwells on an important kinematical question concerning massless 
particles.

In Section \ref{sec:shapely} the shape of the coadjoint orbits and
their stabilizer subgroups are investigated. A surprisingly simple
kinematics is thereby uncovered. We deal as well with the symplectic
structure of those orbits and investigate the covariance properties of
the coordinates under free motion. Section~4 is the Conclusion.

Readers are advised to check our conventions for relativistic
kinematics in Appendix~\ref{app:P-Lie-structure}, before tackling what
follows.

\section{The Kirillov program for the Poincar\'e group}
\label{sec:maskirovka}

The \textit{adjoint action} $\unl{\Ad}$ of $\sP$ on its Lie algebra
$\pl$ we compute as follows. The notation $\ad(X)Y := [X,Y]$ for 
$X,Y \in \pl$ leads to $\Ad(\exp X)Y = e^{\ad(X)}Y 
= Y + [X,Y] + \frac{1}{2!} [X,[X,Y]] + \cdots$. From this one can find
$\Ad(\exp X)Y$ whenever $X = -a^0 H$, $\aaa\.\PP$, $\al \mm\.\LL$ or
$\ze \nn\.\KK$, with $\mm$ and~$\nn$ denoting unit $3$-vectors; here
$Y = H \equiv P^0$, $P^a$, $L^a$ or~$K^a$ are the respective
generators of time and space translations, rotations and boosts.

For instance, if $X = \ze \nn\.\KK$, $Y = H$, then
\begin{align*}
\Ad(\exp(\ze \nn\.\KK)) H &= H + \ze[\nn\.\KK, H] 
+ \frac{\ze^2}{2!}[\nn\.\KK,[\nn\.\KK, H]] 
+ \frac{\ze^3}{3!}[\nn\.\KK, [\nn\.\KK, [\nn\.\KK, H]]] + \cdots 
\\
&= H - \ze \nn\.\PP + \frac{\ze^2}{2!} H
- \frac{\ze^3}{3!} \nn\.\PP + \cdots
= H \cosh\ze - \nn\.\PP \sinh\ze.
\end{align*}
In this way one obtains Table~\ref{tbl:Ad-action}, exhibiting the
adjoint action of~$\sP$ in a perspicuous manner. Explicitly for the
rotation of angle $\al$ around the axis determined by the unit
vector~$\mm$, acting on a vector~$\vv$:
\begin{equation}
R_{\al,\mm}\vv = \vv\,\cos\al + \mm \x \vv \,\sin\al
+ (\mm\cdot\vv) \mm(1 - \cos\al).
\label{eq:widening-gyre} 
\end{equation}

\begin{table}[tbp]
\centering
\caption{The adjoint action $\Ad(\exp X)Y$}
\vspace{6pt}
\begin{tabular}{c|llll}
\hline\hline
$Y \diagdown X$ 
& $-a^0 H$ & $\aaa\.\PP$ & $\al \mm\.\LL$ & $\ze \nn\.\KK$
\rule{0pt}{12pt} 
\\[\jot] 
\hline
$H$ & $H$ & $H$ & $H$ & $H \cosh\ze - \nn\.\PP \sinh\ze$
\rule{0pt}{12pt}
\\[\jot]
$\PP$ & $\PP$ & $\PP$ & $R_{\al,\mm}^{-1} \PP$
& $\PP - H \nn \sinh\ze + (\nn\.\PP) \nn(\cosh\ze - 1)$
\\[\jot]
$\LL$ & $\LL$ & $\LL - \aaa\x\PP$ & $R_{\al,\mm}^{-1} \LL$
& $\LL \cosh\ze - \nn\x\KK \sinh\ze - (\nn\.\LL) \nn (\cosh\ze - 1)$
\\[\jot]
$\KK$ & $\KK - a^0\PP$ & $\KK + H\aaa$ & $R_{\al,\mm}^{-1} \KK$
& $\KK \cosh\ze + \nn\x\LL \sinh\ze - (\nn\.\KK)\nn (\cosh\ze - 1)$
\\[\jot]
\hline\hline
\end{tabular}
\label{tbl:Ad-action} 
\end{table}

\medskip

The \textit{coadjoint action} $\unl{\Coad}$ of $\sP$ on elements of
the Lie coalgebra~$\pl^*$,
$$
\duo<\Coad(\exp X)u, Y> := \duo<u, \Ad(\exp(-X))Y>
\word{for} u \in \pl^*,
$$
can now be derived immediately. Let $h$ be the linear coordinate on
$\pl^*$ associated to~$H$, and similarly let $p^a,l^a,k^a$ be the
coordinates associated to $P^a,L^a,K^a$ ($a = 1,2,3$). The action is
given in these coordinates by Table~\ref{tbl:Coad-action}.

\begin{table}[tbp]
\centering
\caption{The coadjoint action $\Coad(\exp X) y$}
\vspace{6pt}
\begin{tabular}{c|llll}
\hline\hline
$y \diagdown X$ 
& $-a^0 H$ & $\aaa\.\PP$ & $\al \mm\.\LL$ & $\ze \nn\.\KK$
\rule{0pt}{12pt}
\\[\jot]
\hline
$h$ & $h$ & $h$ & $h$ & $h \cosh\ze + \nn\.\pp \sinh\ze$
\rule{0pt}{12pt}
\\[\jot]
$\pp$ & $\pp$ & $\pp$ & $R_{\al,\mm}\,\pp$
& $\pp + h \nn \sinh\ze + (\nn\.\pp) \nn (\cosh\ze - 1)$
\\[\jot]
$\lll$ & $\lll$ & $\lll + \aaa\x\pp$ & $R_{\al,\mm}\,\lll$
& $\lll \cosh\ze + \nn\x\kk \sinh\ze - (\nn\.\lll)\nn (\cosh\ze - 1)$
\\[\jot]
$\kk$ & $\kk + a^0\pp$ & $\kk - h\aaa$ & $R_{\al,\mm}\,\kk$
& $\kk \cosh\ze - \nn\x\lll \sinh\ze - (\nn\.\kk) \nn (\cosh\ze - 1)$
\\[\jot]
\hline\hline
\end{tabular}
\label{tbl:Coad-action} 
\end{table}

We shall need the natural Lie--Poisson bracket on~$\pl^*$: given 
$f \in C^\infty(\pl^*)$, one can regard $df(u)$ as an element of the
Lie algebra, and one obtains:
\begin{equation}
\{f, g\}(u) := \duo<u, [df(u), dg(u)]>
= c^{\al\bt}{}_{\!\ga} \pd{f(u)}{u^\al}\, \pd{g(u)}{u^\bt}\, u^\ga,
\label{eq:Lie-Poisson} 
\end{equation}
where the $c^{\al\bt}{}_{\!\ga}$ are the structure constants of~$\pl$.
Therefore, taking $(h,\pp,\lll,\kk)$ as cartesian coordinates
on~$\pl^*$, their Poisson brackets are given directly by the
commutation relations \eqref{eq:comm-relns} among the corresponding
Lie algebra generators. For reference:%
\footnote{For the problem at hand, the nonrelativistic notation, which
separates rotations from boosts and time from space translations, is
more transparent.}
\begin{alignat}{3}
\{l^a,l^b\} &= \epsi^{ab}_c\, l^c,  \qquad 
& \{l^a,k^b\} &= \epsi^{ab}_c\, k^c,  \qquad
& \{k^a,k^b\} &= -\epsi^{ab}_c\, l^c,
\notag \\
\{l^a, p^b\} &= \epsi^{ab}_c\, p^c,  \qquad 
& \{p^b, k^a\} &= \dl^{ab} h, \qquad
& \{h, k^a\} &= p^a.
\label{eq:poisson-basics} 
\end{alignat}

\subsection{The Casimir functions}
\label{sec:Russian-defense}

The Lie-Poisson bracket~\eqref{eq:Lie-Poisson} restricts to symplectic
structures on the coadjoint orbits foliating it. Generally speaking,
the orbits arise as level sets of two ``Casimir functions'' $C_1$,
$C_2$ on~$\pl^*$ that are invariant by the coadjoint action. These are
easy to obtain explicitly. Let $p = (h,\pp)$ be the energy-momentum
$4$-vector and $w = (w^0,\ww)$ the phase-space ``Pauli--Luba\'nski''
$4$-vector, defined by
$$
w^0 = \lll\.\pp;  \qquad  \ww = \kk \x \pp + h\,\lll.
$$

Clearly $p$ and $w$ are orthogonal in the Minkowski sense: $(pw) = 0$.
{}From Table~\ref{tbl:Coad-action}, with a little work, one verifies
that $(w^0,\ww)$ transforms like $(h,\pp)$ under the coadjoint action.
In particular, under the boost $\Coad(\exp(\ze\nn\.\KK))$:
\begin{align*}
w^0 &\mapsto w^0 \cosh\ze + \nn\.\ww \sinh\ze,
\\
\ww &\mapsto \ww + w^0 \nn \sinh\ze + (\nn\.\ww) \nn (\cosh\ze - 1).
\end{align*}

\begin{lema} 
\label{lm:Poisson-PL}
The Poisson brackets of the components of $w$ with the basic variables
are given by:
\begin{alignat}{3}
\{h, w^\mu\} &= 0,
& \{k^a, w^0\} &= - w^a,
& \{l^a, w^0\} &= 0,
\notag \\
\{p^a, w^\mu\} &= 0, \qquad
& \{k^a, w^b\} &= - \dl^{ab} w^0, \qquad
& \{l^a, w^b\} &= \epsi^{ab}_c w^c;
\label{eq:Poisson-PL} 
\end{alignat}
and among the components, the brackets are:
\begin{equation}
\{w^0, w^a\} = (\ww \x \pp)^a,  \qquad
\{w^a, w^b\} = \epsi^{ab}_c (h w^c - w^0 p^c).
\label{eq:PL-brackets} 
\end{equation}
\end{lema}

\begin{proof}
By direct computation, using~\eqref{eq:poisson-basics}.
\end{proof}

\begin{prop} 
\label{pr:Casimir-oynoveo}
The Casimir functions we seek are
$$
C_1 := (pp) = h^2 - |\pp|^2, \qquad 
C_2 := (ww) = (\lll\.\pp)^2 - |\kk\x\pp + h\lll|^2.
$$
\end{prop}

\begin{proof}
The relations \eqref{eq:Poisson-PL} follow from
\eqref{eq:poisson-basics}. That $\{h,w^0\} = 0$ is clear; the others
are:
\begin{align*}
\{h, w^a\} &= \{h, \epsi^a_{bc} k^b p^c \}
= \epsi^a_{bc} p^b p^c = 0,
\\
\{p^a, w^0\} &= \dl_{bd} \{p^a, l^b\} p^d
= \dl_{bd} \epsi^a_{bc} p^c p^d = 0, 
\\
\{p^a, w^b\} &= \{p^a, \epsi^b_{cd} k^c p^d + h l^b \}
= \epsi^b_{ad} h p^d + \epsi^{ab}_c h p^c = 0,
\\
\{k^a, w^0\} &= \{k^a, \dl_{bd} l^b p^d\} 
= \dl_{bd} (\{k^a, l^b\} p^d + l^b \{k^a, p^d\})
= \epsi^a_{bc} k^c p^b - h l^a = - w^a, 
\\
\{k^a, w^a\} &= \{k^a, \epsi^a_{bd} k^b p^d + h l^a \}
= - \epsi^a_{bd} \epsi^{ab}_c l^c p^d  - l^a p^a
= - \dl_{cd} l^c p^d = - w^0,
\\
\{k^a, w^b\} &= \{k^a, \epsi^b_{cd} k^c p^d + h l^b \}
= - \epsi^b_{cd} \epsi^{ac}_e l^e p^d - \epsi^b_{ca} k^c h 
  - l^b p^a + \epsi^{ab}_c k^c h 
\\
&= l^b p^a - l^b p^a = 0 \quad\text{if } a \neq b,
\\
\{l^a, w^0\} &= \dl_{bd} \{l^a, l^b p^d\} 
= \epsi^a_{bc} (l^c p^b + l^b p^c) = 0,
\\
\{l^a, w^b\} &= \{l^a, (\kk \x \pp)^b\} + h \{l^a, l^b\} 
= \epsi^{ab}_c \bigl( (\kk \x \pp)^c + h l^c \bigr)
= \epsi^{ab}_c w^c;
\end{align*}
From these brackets, \eqref{eq:PL-brackets} follows easily:
\begin{align*}
\{w^0, w^a\} &= \{w^0, \epsi^a_{bd} k^b p^d + h l^a \}
= \epsi^a_{bd} w^b p^d = (\ww \x \pp)^a,
\\
\{w^a, w^b\} &= \{w^a, \epsi^b_{cd} k^c p^d + h l^b \}
= \epsi^b_{cd} \dl^{ac} w^0 p^d + \epsi^{ab}_c h w^c 
= \epsi^{ab}_c (h w^c - w^0 p^c).
\end{align*}
That $C_1$ is a Casimir hardly needs proof. From formulas 
\eqref{eq:Poisson-PL} it follows that
$$
\{k^a, C_2\} = 0 = \{l^a, C_2\};
$$
and Proposition \ref{pr:Casimir-oynoveo} is proved.
\end{proof}

\subsection{Searching for position coordinates}
\label{ssc:taking-position}

In order to find and study the orbits corresponding to the~WP, it is
natural to look for global \textit{position} functions.

Let us quickly review the \textit{massive} case, to better understand
the situation. That is, we restrict ourselves to orbits for which
$C_1 > 0$, writing $C_1 = m^2$ and we decide on
$h = +\sqrt{m^2 + |\pp|^2}$. Note that $C_1 \geq 0$ implies 
$C_2 \leq 0$. Let $\ka := (m,\zero)$ be the vertex of the forward
hyperboloid given by $p^2 = m^2$. Consider the standard Lorentz boost
$L_p$ which takes $(1,\zero)$ to $u := p/m$; its relation to
$\Coad(\exp(\ze\nn\.\KK))$ is given by
$$
\nn := \frac{\uu}{|\uu|} = \frac{\uu}{\sqrt{(u^0)^2 - 1}}; \qquad
u^0 =: \cosh\ze \quad (\ze \geq 0).
$$
One thus obtains
$$
L_p a = \biggl( \frac{ha^0 + \pp\.\aaa}{m}, \,
\aaa + \Bigl( \frac{a^0}{m} + \frac{\pp\.\aaa}{m(m + h)} \Bigr) \pp
\biggr).
$$
Now $0 = (pw) = (L_p^{-1} p \; L_p^{-1} w) = (\ka\,L_p^{-1} w)$. This
means that $L_p^{-1} w = (0, m\sss)$ for some $3$-vector~$\sss$. The
Casimir 
$C_2 = (ww) = (L_p^{-1} w \; L_p^{-1} w) = -m^2|\sss|^2 \leq 0$ is
constant on any orbit; thus, if $C_2 < 0$ we interpret $\sss$ as the
spin vector. From $(0, m\sss) = L_p^{-1} w$ one derives the relation:
\begin{equation}
m\sss = \ww - \frac{w^0}{m + h}\, \pp
= \ww - \frac{(\ww\.\pp)}{h(m + h)}\, \pp.
\label{eq:mass-spin} 
\end{equation}
For fixed $C_1$ and $C_2$, an orbit $\sO_{ms+}$ has been obtained. One
naturally takes as coordinates on it the momenta~$\pp$ and spherical
coordinates arising from~$m\sss$. Three more come from a
\textit{position} triplet $\qq$, given by~\cite{Ganymede}:
\begin{equation}
\qq := - \frac{\kk}{h} - \frac{\pp \x \ww}{mh(m + h)} 
= - \frac{\kk}{h} - \frac{\pp \x \sss}{h(m + h)}\,.
\label{eq:take-a-position} 
\end{equation}
The expressions of the $\pl^*$-coordinates $\lll,\kk$ in terms of the
$\sO_{ms+}$ coordinates $(\qq,\pp,\sss)$ over the orbit are:
$$
\kk = - h\qq - \frac{\pp \x \sss}{m + h}, \quad 
\lll = \qq \x \pp + \sss.
$$
Clearly $\sO_{ms+}$ is homeomorphic to $\bR^6 \x \bS^2$, with isotropy
(or stability) group isomorphic to $\bR \x SO(2)$ -- say, the subgroup
fixing $(\zero,\zero,\sss)$ generated by time translations
$\exp(-a^0 H)$ and rotations $R_{\al,\sss}$, see
\cite[Eqn.~(30)]{Ganymede}.%
\footnote{Coadjoint orbits always have even dimension
\cite{Kirillov04}; and their isotropy groups for maximal dimensional
orbits are always abelian~\cite{DufloV69}.}
The degenerate ``scalar'' case $C_2 = 0$ gives $6$-dimensional orbits
$\sO_{m0+}$, homeomorphic to~$\bR^6$.

Using \eqref{eq:Lie-Poisson} together with the commutation relations,
one verifies that $\set{q^a,p^a}$ are \textit{canonical} coordinates,
and that 
$\{s^a, s^b\} = \epsi^{ab}_c s^c$; $\{s^a, q^b\} = 0 = \{s^a, p^b\}$.
These coordinates, however, are not particularly useful. They do not
transform covariantly for $\sss \neq 0$; in \cite{Ganymede} they were
replaced by others that do~so. They are certainly useless to study the
massless limit. We shall come back to the question of different sets
of position coordinates repeatedly.

\medskip

We turn now to massless particles, the WP in particular. Over fifty
years ago, Wightman wrote a remarkable paper \cite{Wightman62} proving
that quantum spinning massless systems like the photon are not
localizable, in the sense that for them the action of the Euclidean
group cannot be realized on a set of three position coordinates in the
standard way. However, he assumes commutativity of those variables. In
the first chapter of his book~\cite{Schwinger70}, Schwinger disclosed
a view on relativistic position operators that can be understood as a
retort to Wightman's. In fact, allowing for noncommutativity, there
\textit{are} position operators for photons and other fixed-helicity
particles, as well as for the WP, with the correct transformation
properties.

It is both instructive and convenient in our phase-space context to
seek a position vector for the massive case with good limit properties
as $m \downto 0$. Thus we adopt a set of coordinates~$\rr$ suggested
by Schwinger's ideas,%
\footnote{A discussion closer to the original treatment is given in
Appendix~\ref{app:PL-Schwinger}.}
provided~by
\begin{equation}
\rr := - \frac{\kk}{h} + \frac{\pp \x \ww}{h^2(m + h)}\,.
\label{eq:pre-positioning} 
\end{equation}
Notice that as $m \downto 0$, the right hand side of
Eqn.~\eqref{eq:mass-spin} becomes:
\begin{equation}
\ttt := \ww - \frac{(\ww\.\pp)\pp}{h^2} = \ww - \frac{w^0}{h} \pp\,.
\label{eq:tea-for-two} 
\end{equation}
Now we consider the WP case, where by definition $C_1 = m^2 = 0$
and $C_2 = w^2 =: -\rho^2 < 0$. Note that, since $(wp) = 0$, this 
entails $p \nparallel w$ and therefore $\ttt \neq \zero$; also,
$\pp \neq \zero$ since $C_1 = 0$ but $w^2 \neq 0$. Moreover,
$(wp) = 0$ implies that $\ttt \perp \pp$ and $|\ttt|^2 = \rho^2$; and
thus $|(\pp/h) \x \ttt|^2 = \rho^2$, as well.%
\footnote{The parameter $\rho$ has the physical dimension of energy.
Orbits with different values of~$\rho$ correspond to different
particles.}

Position coordinates for the WP are given by
\begin{equation}
\rr := - \frac{\kk}{h} + \frac{\pp \x \ww}{h^3} 
= - \frac{\kk}{h} + \frac{\pp \x \ttt}{h^3} 
= \frac{(\pp\.\kk)\pp}{h^3} + \frac{\pp \x \lll}{h^2}\,.
\label{eq:now-in-position} 
\end{equation}

Introduce the notation $\hel$ for the important \textit{helicity}
variable:
$$
\hel := \frac{\lll \. \pp}h = \frac{w^0}{h}\,.
$$
One readily obtains the basic coalgebra functions $\kk$, $\lll$ in
terms of the new $(\rr,\pp,\hel,\ttt)$ set of variables.%
\footnote{We have seen that the last of these is constrained by
$\ttt\.\pp = 0$ and $\ttt^2 = \rho^2$, so certainly these maximal
orbits are $8$-dimensional.}
\begin{equation}
\ww = \hel\pp + \ttt, \word{thus} 
\kk = -h\,\rr + \frac{\pp \x \ttt}{h^2}, \quad
\lll = h^{-1}(\ww - \kk \x \pp) = \rr \x \pp + \frac{\hel}{h}\,\pp.
\label{eq:keep-apart} 
\end{equation}

We require to have available the Poisson brackets involving the new
variables. Remark first that $\{\hel, p^a\} = \{\hel, l^a\} = 0$, and
that the first relation in~\eqref{eq:PL-brackets} can be rewritten as:
$$
\{w^0, t^a\} = \{w^0, w^a\} = - (\pp \x \ww)^a = - (\pp \x \ttt)^a.
$$

\begin{lema} 
\label{lm:tea-table}
The helicity and the $3$-vector $\ttt$ have these Poisson brackets:
\begin{alignat}{2}
\{\hel, k^a\} &= t^a/h, 
& \{k^a, t^b\} &= t^a p^b/h, 
\notag \\
\{\hel, \ttt\} &= - \pp/h \x \ttt, \qquad
& \{l^a, t^b\} &= \epsi^{ab}_c\, t^c,
\notag \\
\{\hel, \pp/h \x \ttt\} &= \ttt,
& \{p^a, t^b\} &= \{t^a, t^b\} = 0.
\label{eq:tea-table} 
\end{alignat}
\end{lema}

\begin{proof}
The first relation comes from
$$
\{\hel, k^a\} = \{w^0 h^{-1}, k^a\} 
= \{w^0, k^a\} h^{-1} - w^0 h^{-1} \{h, k^a\} h^{-1}
= (w^a - \hel\,p^a)/h = t^a/h.
$$
The second one is given by
$$
\{k^a, t^b\} = \{k^a, w^b - w^0 p^b h^{-1}\} 
= -\dl^{ab} w^0 + w^a p^b h^{-1} + \dl^{ab} w^0 - w^0 p^b p^a h^{-2}
= t^a p^b/h.
$$
Furthermore,
\begin{align*}
\{\hel, \ttt\} &= \{\hel, \ww\} = \{\hel, \kk \x \pp\}
= \{\hel, \kk\} \x \pp = - \pp/h \x \ttt,
\\
\shortintertext{and thus also}
\{\hel, \pp/h \x \ttt\} &= \pp/h \x \{\hel, \ttt\}
= - \pp \x (\pp \x \ttt)/h^2 = (|\pp|^2/h^2) \,\ttt = \ttt,
\end{align*}
on account of $(pp) = 0$ and $\pp\.\ttt = 0$. From equations
\eqref{eq:poisson-basics} and \eqref{eq:Poisson-PL}, we see at once
that $\{l^a, t^b\} = \epsi^{ab}_c\,t^c$. Similarly, $\{p^a, t^b\} = 0$
follows from \eqref{eq:Poisson-PL} and~\eqref{eq:tea-for-two}.

Note that the second relation of~\eqref{eq:PL-brackets} can now be
shortened to $\{w^a, w^b\} = \epsi^{ab}_c\, h t^c$. In consequence,
the components of~$\ttt$ Poisson-commute:
\begin{align*}
\{t^a, t^b\} 
&= \{w^a, w^b\} - p^a \{\hel, w^b\} - \{w^a, \hel\} p^b
\\
&= \epsi^{ab}_c\,ht^c + p^a (\pp/h \x \ttt)^b - (\pp/h \x \ttt)^a p^b
\\
&= \epsi^{ab}_c \bigl( ht^c + (\pp \x (\pp \x \ttt))^c/h \bigr) = 0.
\tag*{\qed}
\end{align*}
\hideqed
\end{proof}

\begin{lema} 
\label{lm:taking-position}
The Poisson brackets involving the position variables are given by:
\begin{alignat}{4}
\{r^a, p^b\} &= \dl^{ab}, \quad
& \{r^a, h\} &= p^a/h, 
& \{\rr, w^0\} &= \ww_\parallel/h,
& \{r^a, \hel\} &= 0,
\notag \\
\{r^a, t^b\} &= - t^a p^b/h^2, \quad
& \{l^a, r^b\} &= \epsi^{ab}_c \, r^c, \quad
& \{r^a, r^b\} &= - \epsi^{ab}_c \,\hel p^c/h^3,
\label{eq:taking-position} 
\end{alignat}
and also:
$$
\{r^a, w^b\} = \la \dl^{ab} - t^a p^b/h^2, \quad
\{r^a, k^b\} = - p^a r^b/h + \epsi^{ab}_c(\la p^c - t^c)/h^2
- 2 p^a(\pp \x \ttt)^b/h^4.
$$
\end{lema}

\begin{proof}
A few of these follow directly from the basic Poisson brackets
\eqref{eq:poisson-basics}. Indeed, since $\{h, p^b\} = 0$ and
$\{(\pp \x \ww)^a, p^b\} = 0$, the first reduces to
$\{r^a, p^b\} = -\{k^a, p^b\}/h = \dl^{ab}$. That is to say,
$\rr$ and~$\pp$ are conjugates, as is naturally demanded of a position
vector.

In like manner, $\{r^a, h\} = - \{k^a/h, h\} = p^a/h$. Also from
\eqref{eq:poisson-basics} we get:
$$
\{l^a, r^b\} = - \{l^a, k^b\}/h + \{l^a, (\pp \x \ww)^b\}/h^3
= \epsi^{ab}_c \bigl( - k^c/h + (\pp \x \ww)^c/h^3 \bigr)
= \epsi^{ab}_c \, r^c.
$$

Three more relations follow from~\eqref{eq:tea-table}:
\begin{align*}
\{\rr, w^0\} &= - \{\kk, w^0\}/h + \{\pp \x \ttt, w^0\}/h^3
= \ww/h + \pp \x (\pp \x \ttt)/h^3 = \ww/h - \ttt/h = \ww_\parallel/h,
\\
\{\rr, \hel\} &= - \{\kk, \hel\}/h + \{\pp/h \x \ttt, \hel\}/h^2
= - \ttt/h^2 + \ttt/h^2 = \zero,
\\
\{r^a, t^b\} &= - \{k^a/h, t^b\} = - t^a p^b/h^2.
\end{align*}
These in turn imply that
$$
\{r^a, w^b\} = \{r^a, \la p^b + t^b\} 
= \{r^a, \la\} p^b + \la \{r^a, p^b\} + \{r^a, t^b\}
= \la \dl^{ab} - t^a p^b/h^2.
$$

The last relation in~\eqref{eq:taking-position} becomes a sum of three
terms:
\begin{align*}
\{r^a, r^b\} &= \{k^a/h, k^b/h\} - \{k^a/h, (\pp \x \ttt)^b/h^3\}
- \{(\pp \x \ttt)^a/h^3, k^b/h\},
\\
\shortintertext{where}
\{k^a/h, k^b/h\}
&= \{k^a, k^b\} h^{-2} + k^a \{h^{-1}, k^b\} h^{-1}
+ k^b \{k^a, h^{-1}\} h^{-1}
\\
&= - \epsi^{ab}_c\, l^c/h^2 - k^a p^b/h^3 + k^b p^a/h^3
= - \epsi^{ab}_c\, w^c/h^3;
\\
- \{k^a/h, (\pp\x\ttt)^b/h^3\}
&= \epsi^b_{de} \{k^a, t^d p^e/h^3\}/h
\\
&= \epsi^b_{de} \bigl( \{k^a, t^d\} p^e/h^4 + t^d \{k^a, p^e\}/h^4 
- 3 t^d p^e \{k^a, h\}/h^5 \bigr)
\\
&= \epsi^b_{de} \bigl( t^a p^d p^e/h^4 - \dl^{ae} t^d/h^3
+ 3 t^d p^e p^a/h^5 \bigr)
\\
&= - \epsi^{ab}_c\, t^c/h^3 - 3(\pp \x \ttt)^b p^a/h^5;
\end{align*}
and similarly,
$- \{(\pp \x \ttt)^a/h^3, k^b/h\}
= - \epsi^{ab}_c\, t^c/h^3 - 3(\pp \x \ttt)^a p^b/h^5$. Therefore,
$$
\{r^a, r^b\} 
= - \epsi^{ab}_c \bigl( h^2 w^c + 2 h^2 t^c 
+ 3(\pp \x (\pp \x \ttt))^c \bigr)/h^5
= \epsi^{ab}_c (t^c - w^c)/h^3 
= - \epsi^{ab}_c \, \hel p^c/h^3.
$$

Using Eqs.~\eqref{eq:keep-apart} and~\eqref{eq:taking-position}, it
now follows that
\begin{align*}
\{r^a, k^b\} &= -\{r^a, h r^b\} + \{r^a, \epsi^b_{de}\, p^d t^e/h^2\}
\\
&= - \{r^a, h\} r^b - h \{r^a, r^b\} 
+ \epsi^b_{de} \bigl( \{r^a, p^d\} t^e/h^2 + p^d \{r^a, t^e\}/h^2
- 2 p^d t^e \{r^a, h\}/h^3 \bigr)
\\
&= - p^a r^b/h + \epsi^{ab}_c(\la p^c - t^c)/h^2
- 2 p^a(\pp \x \ttt)^b/h^4.
\tag*{\qed}
\end{align*}
\hideqed
\end{proof}

The formulas of the previous lemmata remain valid for the
fixed-helicity situations, by replacing $\ttt$ by~$\zero$ and
``freezing'' $\hel$ to a given value. The commutation relations
$\{r^a, r^b\} = - \epsi^{ab}_c\,\hel p^c/h^3$ were actually found for
the latter situation already in~\cite{BalachandranMSSZ92}.

\subsection{On the Wigner rotation}
\label{ssc:gira-el-mundo}

Let us return briefly to the massive case, $m > 0$. Let $w' = \La w$,
and consider accordingly
$$
(0,m\sss') := L_{\La p}^{-1} w'= L_{\La p}^{-1} \La L_p(0,m\sss).
$$
This transformation is just the Wigner rotation $g(\La,p)$. The spin's
axis of rotation is given by $\pp \x \nn$ for a boost $\La$ in the
direction of~$\nn$: when the boost is parallel to the momentum $\pp$,
the Wigner rotation is trivial. With
$\mm := (\pp \x \nn)/|\pp\x\nn|$, the formula is~\cite{Ganymede}:
\begin{align}
\sss'= g(\La,p) \sss &= R_{\dl,\mm} \sss 
= \sss + (\mm \x \sss) \sin\dl
- \bigl( \sss - (\mm\.\sss)\mm \bigr)(1 - \cos\dl),
\label{eq:spin-spun} 
\\
\shortintertext{with angle~$\dl$ given by}
\sin\dl &= \frac{(m + h) \sinh\ze + (\nn\.\pp)(\cosh\ze - 1)}
{(m + h)(m + h')}\, |\pp\x\nn|,
\label{eq:Wigner-angle} 
\end{align}
where $h' = h \cosh\ze + \nn\.\pp \sinh\ze$ by
Table~\ref{tbl:Coad-action}.

Under a boost in the direction of $\nn$, the momentum also turns
around $\pp \x \nn$. This is true in all generality: from the
coadjoint action for boosts -- see Table~\ref{tbl:Coad-action} -- with
$p' = \La p$, we obtain
$$ 
\pp \x \pp' 
= \bigl[ h\sinh\ze + (\nn\.\pp) (\cosh\ze - 1) \bigr]\, \pp \x \nn.
$$
Therefore, the component of $\pp'$ not along $\pp$ stays on the plane
perpendicular to $\pp \x \nn$. The sine of the rotation angle is given
by
\begin{equation}
\frac{|\pp \x \pp'|}{|\pp||\pp'|} 
= \frac{h\sinh\ze + (\nn\.\pp)(\cosh\ze - 1)}{|\pp||\pp'|}\,
|\pp \x \nn|.
\label{eq:rotation-angle} 
\end{equation}
This is in general greater than the Wigner angle given
by~\eqref{eq:Wigner-angle}. Now comes a key point: although not all
the factors in its definition do~so, the Wigner rotation formula
itself makes perfect sense for $m = 0$. Namely, keeping in mind that
in this case $h = |\pp|$ and $h' = |\pp'|$,
formula~\eqref{eq:Wigner-angle} then perfectly matches
formula~\eqref{eq:rotation-angle}.

For good measure, we give next the brute-force proof that rotating
$\pp$ around $\pp\x\nn$ with rotation angle given by the massless
limit of \eqref{eq:Wigner-angle} yields the expected swing from a
boost on~$\pp$. We shall also need that
\begin{equation}
\sin\dl = \frac{h \sinh\ze + (\nn\.\pp)(\cosh\ze - 1)}{hh'}\,
|\pp\x\nn| \,, \qquad
\cos\dl = 1 - \frac{|\pp \x \nn|^2}{hh'}\,(\cosh\ze - 1),
\label{eq:rotation-angle-bis} 
\end{equation}
in that limit. (Using $|\pp \x \nn|^2 = h^2 - (\nn \. \pp)^2$, these
expressions yield $\cos^2\dl + \sin^2\dl = 1$, so they do define an
angle~$\dl$.) Now, since the axis of rotation $\mm$ is perpendicular
to~$\pp$, one finds from Eqn.~\eqref{eq:spin-spun} that
\begin{align}
R_{\dl,\mm}(\pp) &= \pp \cos\dl + (\mm\x\pp) \sin\dl = \pp \cos\dl +
(h^2\nn - (\nn \. \pp)\pp)\,\frac{\sin\dl}{|\pp \x \nn|} \,,
\notag \\
&= \pp - \biggl\{ \frac{\cosh\ze - 1}{hh'}\,(h^2 - (\nn \. \pp)^2)
+ \frac{\nn \. \pp}{hh'} \bigl( h \sinh\ze 
+ \nn \. \pp\,(\cosh\ze - 1) \bigr) \biggr\} \pp 
\notag \\
&\qquad + \frac{h}{h'} \bigl( h \sinh\ze 
+ \nn \. \pp\,(\cosh\ze - 1) \bigr) \nn
\notag \\
&= \pp 
- \frac{1}{h'} \bigl( h(\cosh\ze - 1) + \nn \. \pp \sinh\ze \bigr) \pp
+ \frac{h}{h'} \bigl( h \sinh\ze + \nn \. \pp(\cosh\ze - 1) \bigr) \nn
\notag \\
&= \pp - \frac{h' - h}{h'} \pp + \frac{h}{h'} \bigl( h \sinh\ze + \nn
\. \pp \,(\cosh\ze - 1) \bigr) \nn
\notag \\
&= \frac{h}{h'} \bigl( \pp + h \nn \sinh\ze 
+ (\nn \. \pp)\nn (\cosh\ze - 1) \bigr) = \frac{h}{h'} \pp'.
\label{eq:roll-pointer} 
\end{align}
Therefore, $\pp'/h' = R_{\dl,\mm}(\pp/h)$, with
$\mm = (\pp\x\nn)/|\pp\x\nn|$ and $\dl$ given 
by~\eqref{eq:rotation-angle-bis}.

\section{The shape of the orbits}
\label{sec:shapely}

A nagging worry for some readers may have been that, contrary to the
massive case, there is no distinguished point for the momentum in the
orbit of a WP, nor there is a continuous cross-section of the Lorentz
principal bundle over it~\cite{BoyaCS74}. Fortunately, the kinematics
of the WP saves the day.

\subsection{Coordinate transformations}
\label{ssc:on-the-move}

In order to understand the structure of the orbits corresponding to
Wigner particles, we need to examine the effect of boosts on
$\hel,\ttt,\rr$, which perhaps is not obvious \textit{a priori}.

As previously indicated, the $\Coad$ formulas for $(h,\pp)$ are good
for $(w^0,\ww)$. With $\nn$ the direction of a boost and $\ze$ its
parameter, we therefore obtain for the helicity (an invariant under
translations and rotations):
$$
\hel = \frac{w^0}{h} \longmapsto \frac{(w^0)'}{h'} 
\equiv \frac{w^0 \cosh\ze + \nn\.\ww \sinh\ze}
{h \cosh\ze + \nn\.\pp \sinh\ze} 
= \hel + \frac{\nn\.\ttt \tanh\ze}{h + \nn\.\pp \tanh\ze}.
$$
In particular $\hel$ is invariant in the case $\nn = \pp/h$. 

\medskip

Next, $\ww - w^0\pp/h = \ttt \mapsto \ttt'$, where
\begin{align}
\ttt' &:= \ttt + (\nn\.\ttt)\nn (\cosh\ze - 1) 
- \frac{\nn\.\ttt \tanh\ze}{h + \nn\.\pp \tanh\ze} 
\bigl( \pp + h \nn \sinh\ze + (\nn\.\pp)\nn (\cosh\ze - 1) \bigr)
\notag \\
&= \ttt + \nn\.\ttt \bigl( (\cosh\ze - 1)\nn - \sinh\ze\pp'/h' \bigr).
\label{eq:tea-is-served} 
\end{align}
The expression \eqref{eq:tea-is-served} is linear in $\ttt$, with
coefficients depending on $\pp,\nn,\ze$. We verify it:
\begin{align*}
\ttt' = \frac{h'\ww' - (w^0)'\pp'}{h'}
&= \frac{1}{h'} \bigl[ (h \cosh\ze + \nn\.\pp \sinh\ze) 
\bigr( \ww + w^0\nn \sinh\ze + (\nn\.\ww)\nn(\cosh\ze - 1) \bigr) 
\\
&\quad - (w^0 \cosh\ze + \nn\.\ww \sinh\ze) 
\bigl( \pp + h\nn \sinh\ze + (\nn\.\pp)\nn(\cosh\ze - 1) \bigr)\bigr].
\end{align*}
The computation proceeds by systematically cancelling all terms
in~$w^0$. There remains
\begin{align*}
\ttt' &= \frac{1}{h'} \bigl[ 
h \cosh\ze \bigl( \ttt + (\nn\.\ttt)\nn (\cosh\ze - 1) \bigr)
- \bigl( h (\nn\.\ttt)\nn \sinh\ze - (\nn\.\pp)\ttt 
+ (\nn\.\ttt) \pp \bigr) \sinh\ze \bigr]
\\
&= \ttt + (\nn\.\ttt)\nn (\cosh\ze - 1) 
\\
&\qquad - \frac{1}{h'} \bigl[ 
h (\nn\.\ttt)\nn \sinh^2\ze + (\nn\.\ttt) \pp \sinh\ze
+ (\nn\.\pp) (\nn\.\ttt)\nn \sinh\ze(\cosh\ze - 1) \bigr]
\\
&= \ttt + (\nn\.\ttt)\nn (\cosh\ze - 1) 
- \frac{(\nn\.\ttt) \sinh\ze}{h'} \bigl[ \pp + h \nn \sinh\ze 
+ (\nn\.\pp) \nn (\cosh\ze - 1) \bigr].
\\
&= \ttt + (\nn\.\ttt)\bigl[(\cosh\ze - 1)\nn - \sinh\ze\pp'/h'\bigr].
\end{align*}

\begin{thm} 
\label{thm:moving-frame}
The mapping $\ttt \mapsto \ttt'$ is implemented by the \textbf{same}
rotation $R_{\dl,\mm}$ of formula~\eqref{eq:roll-pointer}, with axis
$\mm := (\pp \x \nn)/|\pp \x \nn|$ and angle $\dl$ given by
\eqref{eq:rotation-angle-bis}. For $\nn = \pm\pp/h$, the rotation is
trivial and $\ttt' = \ttt$.
\end{thm}

\begin{proof}
To check the equality $\ttt' = R_{\dl,\mm}(\ttt)$, it is enough to
show that these two $3$-vectors have the same components with respect
to some $3$-vector basis. For that purpose, choose the orthogonal
moving frame $\{\pp/h, \ttt, \pp/h \x \ttt\}$. We claim that
\begin{align}
(\pp/h) \. \ttt' = (\pp/h) \. R_{\dl,\mm}\,\ttt, \quad
\ttt \. \ttt' = \ttt \. R_{\dl,\mm}\,\ttt, \quad
(\pp/h \x \ttt) \. \ttt' = (\pp/h \x \ttt) \. R_{\dl,\mm}\,\ttt.
\label{eq:gyro-scopic} 
\end{align}

First, of all,
\begin{align*}
(\pp/h) \. \ttt' &= (\nn\.\ttt)(\nn\.\pp) \frac{\cosh\ze - 1}{h} 
- \frac{(\nn\.\ttt) \sinh\ze}{hh'} \bigl[ h^2 + h \nn\.\pp \sinh\ze
+ (\nn\.\pp)^2 (\cosh\ze - 1) \bigr]
\\
&= (\nn\.\ttt)(\nn\.\pp) \frac{\cosh\ze - 1}{hh'} \bigl( 
h' - h(\cosh\ze + 1) - \nn\.\pp \sinh\ze \bigr) 
- \frac{h}{h'} \nn\.\ttt \sinh\ze 
\\
&= - \frac{\nn\.\ttt}{h'} 
\bigl( h \sinh\ze + \nn\.\pp (\cosh\ze - 1) \bigr)
= - \frac{h\,\nn\.\ttt}{|\pp \x \nn|} \sin\dl 
\\
&= (\pp/h) \. (\mm \x \ttt) \sin\dl = (\pp/h) \. R_{\dl,\mm}\,\ttt.
\end{align*}

Next,
\begin{align*}
\ttt \. \ttt' &= \rho^2 + \frac{(\nn\.\ttt)^2}{h'} 
\bigl[ h'(\cosh\ze - 1) - h \sinh^2\ze 
- (\nn\.\pp) \sinh\ze(\cosh\ze - 1) \bigr]
\\
&= \rho^2 + \frac{(\nn\.\ttt)^2}{h'} 
\bigl( h \cosh\ze(\cosh\ze - 1) - h \sinh^2\ze \bigr)
= \rho^2 - \frac{h(\nn\.\ttt)^2}{h'} (\cosh\ze - 1),
\end{align*}
whereas
\begin{align*}
\ttt \. R_{\dl,\mm}\,\ttt 
&= \rho^2 - \bigl( \rho^2 - |\mm \x \ttt|^2 \bigr) (1 - \cos\dl)
= \rho^2 
- \frac{|(\pp \x \nn) \x \ttt|^2}{|\pp \x \nn|^2} (1 - \cos\dl)
\\
&= \rho^2 - \frac{h(\nn\.\ttt)^2}{h'} (\cosh\ze - 1),
\end{align*}
as claimed. We leave the proof of the third equation in
\eqref{eq:gyro-scopic} to the reader.
\end{proof}

Consider boosts along the special directions $\nn = \pp/h$,
$\ttt/\rho$ and $(\pp \x \ttt)/h\rho$. 
(a)~For $\nn = \pp/h$, we get $\ttt' = \ttt$ and $\pp'/h' = \pp/h$:
trivial rotation. 
(b)~For $\nn = \ttt/\rho$, we get $h' = h \cosh\ze$, 
$\mm = (\pp \x \ttt)/h\rho$, and
$$
\pp'/h' = (\pp/h) \sech\ze + (\ttt/\rho) \tanh\ze, \qquad
\ttt'/\rho = - (\pp/h) \tanh\ze + (\ttt/\rho) \sech\ze.
$$
(c)~For $\nn = (\pp \x \ttt)/h\rho$, again $h' = h \cosh\ze$ and
$\pp'/h' = (\pp/h) \sech\ze + (\ttt/\rho) \tanh\ze$, but now 
$\mm = -\ttt/\rho$ and $\ttt' = \ttt$. In case~(b) unsurprisingly 
there holds $\sin\dl = \tanh\ze$ and $\cos\dl = \sech\ze$.

\begin{corl} 
\label{cr:gyro-scope}
Under the Lorentz group action, the moving frame rotates as a
gyroscope:
$$
\pp/h \mapsto R_{\dl,\mm}(\pp/h), \qquad
\ttt \mapsto R_{\dl,\mm}\,\ttt, \qquad
\pp/h \x \ttt \mapsto R_{\dl,\mm}(\pp/h \x \ttt),
$$
where $\mm = (\pp \x \nn)/|\pp \x \nn|$.
\end{corl}

\begin{proof}
The asserted result being true for all rotations and boosts, it is
\textit{ipso facto} true for all Lorentz transformations.
\end{proof}

We deem quite noteworthy this remarkable kinematical behaviour of
the~WP. It has as a consequence Wigner's original equations of
motion~\cite{Wigner47} in a first-quantized
formulation~\cite{Eunike}.

\begin{table}[tbp]
\centering
\caption{The coadjoint action on orbital coordinates}
\vspace{6pt}
\begin{tabular}{c|cccl}
\hline\hline
$u \diagdown X$ 
& $-a^0 H$ & $\aaa\.\PP$ & $\al \mm\.\LL$ & $\ze \nn\.\KK$
\rule{0pt}{12pt}
\\[\jot]
\hline
$\hel$ & $\hel$ & $\hel$ & $\hel$ & 
$\hel + \nn\.\ttt \tanh\ze/(h + \nn\.\pp \tanh\ze)$
\rule{0pt}{12pt}
\\[\jot]
$\pp$ & $\pp$ & $\pp$ & $R_{\al,\mm}\,\pp$
& $\pp + h \nn \sinh\ze + (\nn\.\pp) \nn (\cosh\ze - 1)$
\\[\jot]
$\rr$ & $\rr - a^0 \dfrac{\pp}{h}$ & $\rr + \aaa$ & $R_{\al,\mm}\,\rr$
& $\rr - \dfrac{\nn\.\rr}{h'} (\pp \sinh\ze + h\nn(\cosh\ze - 1))$
\\[2\jot]
&&&&\quad 
$+ \dfrac{\hel}{hh'}\,(\nn \x \pp) \sinh\ze + (\ttt$-dependent~term)
\\[2\jot]
$\ttt$ & $\ttt$ & $\ttt$ & $R_{\al,\mm}\,\ttt$ 
& $R_{\dl,\uu}\,\ttt \ \bigl( 
\text{with $\dl$ as given in \eqref{eq:rotation-angle-bis} and } 
\uu = \frac{\pp\x\nn}{|\pp\x\nn|} \bigr)$
\\[\jot]
\hline\hline
\end{tabular}
\label{tbl:Coad-orbital} 
\end{table}

We now examine the position coordinates: 
$-\kk/h + (\pp \x \ttt)/h^3 = \rr \mapsto \rr'$, where
\begin{align}
\rr' &:= \frac{-\kk\cosh\ze + (\nn\.\kk)\nn(\cosh\ze - 1)}{h'}
+ \frac{\nn\x\lll \sinh\ze}{h'} 
+ \frac{R_{\dl,(\pp\x\nn)/|\pp\x\nn|}(\pp/h \x \ttt)}{h'^2}
\notag \\
&= \frac{h\rr \cosh\ze - h(\nn\.\rr)\nn (\cosh\ze - 1)}{h'}
+ \frac{((\nn\.\pp) \rr - (\nn\.\rr) \pp) \sinh\ze}{h'}
+ \frac{\hel}{hh'}\,(\nn \x \pp) \sinh\ze
\notag \\
&\quad - \frac{\pp \x \ttt \cosh\ze}{h^2h'}
+ \frac{[\nn,\pp,\ttt] \nn (\cosh\ze - 1)}{h^2h'}
+ \frac{R_{\dl,(\pp\x\nn)/|\pp\x\nn|}(\pp/h \x \ttt)}{h'^2}
\notag \\
&= \rr - \frac{\nn\.\rr}{h'} (\pp \sinh\ze + h\nn(\cosh\ze - 1))
+ \text{$(\hel,\ttt)$-dependent terms}.
\label{eq:hardened-position} 
\end{align}

The first three terms in \eqref{eq:hardened-position}, free of the
internal variables $\hel,\ttt$, look different from the transformation
rule for momentum; however, we shall soon see that they make
relativistic sense.

The action of the Poincar\'e group generators now follows from
\eqref{eq:tea-is-served} and~\eqref{eq:hardened-position}. They are
given in Table~\ref{tbl:Coad-orbital}.

For good measure, the infinitesimal actions are also given in
Table~\ref{tbl:coad-infml}.%
\footnote{The information is already contained in the Poisson
brackets, but it is good to cross-check them with the outcomes in
Table~\ref{tbl:Coad-orbital}.}
It helps to note that this action on $w$ follows the pattern of its
action on~$p$: namely, $(\ze \nn\.\KK) \lt h = \ze \nn\.\pp$ whereas
$(\ze \nn\.\KK) \lt w^0 = \ze \nn\.\ww$; $(\al \mm\.\LL) \lt w^0 = 0$;
$(\ze \nn\.\KK) \lt \ww = w^0 \ze \nn$ and 
$(\al \mm\.\LL) \lt \ww = \al\,\mm \x \ww$; and the other generators
act trivially on $w^0$ and~$\ww$.

The $(3,4)$-entry in Table~\ref{tbl:coad-infml} is found by expanding 
the right hand side of \eqref{eq:hardened-position} in powers 
of~$\ze$, using $\sinh\ze = O(\ze)$, $\cosh\ze - 1 = O(\ze^2)$, 
$\cos\dl = 1 + O(\ze^2)$, 
$(\pp\x\nn)\sinh\dl/|\pp\x\nn| = - \ze\nn \x \pp/h + O(\ze^2)$, and
$h/h' = 1 - \ze\nn\.\pp/h + O(\ze^2)$. Therefore:
\begin{align*}
\rr' &= \rr - (\ze\nn\.\pp/h) \rr 
+ \frac{(\ze\nn\.\pp) \rr - (\ze\nn\.\rr) \pp}{h}
+ \frac{\hel}{h^2}\,\ze\nn \x \pp - \frac{\pp \x \ttt}{h^3}
+ \frac{(\ze\nn\.\pp) \pp \x \ttt}{h^4}
\\
&\qquad + \frac{\pp/h \x \ttt}{h^2} 
- 2 (\ze\nn\.\pp/h)\,\frac{\pp/h \x \ttt}{h^2}
- \frac{(\ze\nn \x \pp/h) \x (\pp/h \x \ttt)}{h^2} + O(\ze^2)
\\
&= \rr - (\ze\nn\.\rr) \pp/h + \hel\,\ze\nn \x \pp/h^2
- \frac{(\ze\nn\.\pp) \pp \x \ttt + \ze\nn\.(\pp \x \ttt) \pp}{h^4}
+ O(\ze^2)
\\
&= \rr - (\ze\nn\.\rr) \pp/h + \hel\,\ze\nn \x \pp/h^2
- \ze\nn \x ((\pp \x \ttt) \x \pp)/h^4
- 2 \ze\nn\.(\pp \x \ttt)\,\pp/h^4 + O(\ze^2)
\\
&= \rr - (\ze\nn\.\rr) \pp/h + \ze\nn \x (\hel\pp - \ttt)/h^2
- 2 \ze\nn\.(\pp \x \ttt)\,\pp/h^4 + O(\ze^2).
\end{align*}

The shape of the orbit can be determined already: clearly $\rr$ takes
values in $\bR^3$, then $\pp$ takes values in 
$\bR^3 \less \{\zero\} \approx \bR \x \bS^2\,$. Then
$\hel \in(-\infty,\infty)$ and $\ttt$ takes values on a circle.
Therefore%
\footnote{There is now a coadjoint orbit that is not simply connected.
For non-simply connected coadjoint orbits it is hard to push forward
the Kirillov paradigm -- the known examples such as~\cite{AliKM00}
correspond to groups with trivial stability subgroups -- towards the
derivation of the unitary irreducible representations of the group.}
\begin{equation}
\sO_{m=0,\rho} \approx \bR^3 \x (\bR \x \bS^2) \x (\bR \x \bS^1).
\label{eq:unweeded-garden} 
\end{equation}

\begin{table}[tbp]
\centering
\caption{The infinitesimal coadjoint action $X \lt u = \coad(X) u$}
\vspace{6pt}
\begin{tabular}{c|cccl}
\hline\hline
$u \diagdown X$ 
& $-a^0 H$ & $\aaa\.\PP$ & $\al \mm\.\LL$ & $\ze \nn\.\KK$
\rule{0pt}{12pt}
\\[\jot]
\hline
$\hel$ & $0$ & $0$ & $0$ & $\ze\nn\.\ttt/h$
\rule{0pt}{12pt}
\\[\jot]
$\pp$ & $\zero$ & $\zero$ & $\al\,\mm \x \pp$ & $h \ze\nn$
\\[\jot]
$\rr$ & $- a^0 \pp/h$ & $\aaa$ & $\al\,\mm \x \rr$
& $- (\ze\nn\.\rr)\,\pp/h + \ze\nn \x (\hel\pp - \ttt)/h^2
- 2 \ze\nn\.(\pp \x \ttt)\,\pp/h^4$
\\[2\jot]
$\ttt$ & $\zero$ & $\zero$ & $\al\,\mm \x \ttt$ 
& $-(\ze\nn\.\ttt) \pp/h$
\\[\jot]
\hline\hline
\end{tabular}
\label{tbl:coad-infml} 
\end{table}

\subsubsection{On the stability subgroup}
\label{sss:on-the-rest}

Choose any point $u = (\hel,\pp,\rr,\ttt)$ in $\pl^*$ subject to the
requirements $|\pp|^2 = h^2$ and $|\ttt|^2 = \rho^2 > 0$. To study its
coadjoint orbit under the Poincar\'e group, it is also instructive to
determine the isotropy subgroup $\sP_u$, since the orbit is just the
homogeneous space $\sP/\sP_u$.

The isotropy subgroups at different points on the orbit are conjugate,
so to find a ``representative'' isotropy group we choose a point where
$\rr = \zero$ and $\hel = 0$. We tackle Table~\ref{tbl:Coad-orbital}
one row at a time. From invariance of~$\hel$, excluding the case 
$\ze = 0$, we get $\nn\.\ttt = 0$ for the boost component of an
element of the stability subgroup. Therefore
\begin{equation}
\nn = a\pp/h + b(\pp \x \ttt)/\rho h, \word{with} a^2 + b^2 = 1,
\label{eq:tiny-turn} 
\end{equation}
and $\uu \equiv (\pp \x \nn)/|\pp \x \nn| = \ttt/\rho$. Thus a boost
leaving $\hel$ invariant will leave $\ttt$ invariant as well. That is
indeed so, since then $R_{\dl,\uu}$ is a rotation around the direction
of $\ttt$~itself. The $R_{\al,\mm}$ component of that element
fixes~$\ttt$, too, so $\mm = \uu$. From the second row, we learn that
$a$ in equation~\eqref{eq:tiny-turn} equals $(1 - \cosh\ze)/\sinh\ze$
(necessarily $< 0$), in order to keep $|\pp|$ constant, whereby
$h' = h$ in this case. Since $\pp$ rotates around $\ttt$, this can be
compensated by an ordinary rotation around the same axis, with the
same rotation angle in the opposite direction (i.e., 
$\al\mm = -\dl\uu$). Finally, from the third row, the nonvanishing
terms of a boost action on $\rr = \zero$ and $\hel = 0$ produce
components along $\pp$ and $\pp\x\ttt$, which can be compensated by a
one-dimensional family of choices of $a^0$, and a suitable choice of
$\aaa(\ze)$ -- we need not give its complicated formula here. Notice
that in the present instance no purely-Lorentz solutions can be 
found.%
\footnote{There is a one-dimensional subspace of translations that
acts trivially: when $\ze = 0$, one can take $\al = 0$ as well, and
Table~\ref{tbl:Coad-orbital} shows that the condition 
$\aaa = a^0(\pp/h)$ yields invariance of~$u$ under the coadjoint
action. This is identical to what would occur for massive particles.}
The isotropy subgroup, freely parametrized by $a^0$ and~$\ze$, has the
topology of the plane~$\bR^2$.

\subsubsection{Simultaneity hyperplanes}
\label{sss:untoward}

The $r^a$ coordinates are not canonical, since they do not commute
among themselves. We next examine whether they are relativistically
covariant. Certainly they are covariant under Euclidean
transformations. The argument that follows, taken
from~\cite[Ch.~20]{SudarshanM74}, allows to understand the rule of
change when going from one Lorentz frame to another (say with primed
coordinates) by a boost. In a Hamiltonian formulation, the position
coordinates should be regarded as \textit{initial conditions} for free
motion. Thus consider $\rr(t) = \rr + (\pp/h)\,t$ and likewise
$\rr'(t') = \rr' + (\pp'/h')\,t'$, and \textit{assume} the standard
transformation rules under boosts:
\begin{equation}
t' = t \cosh\ze + \nn\.\rr(t) \sinh\ze; \quad 
\rr'(t') = \rr(t) + t \nn \sinh\ze + (\nn\.\rr(t)) \nn (\cosh\ze - 1).
\label{eq:run-along} 
\end{equation}
We want to examine the resulting relation between 
$\rr'(t' = 0) \equiv \rr'$ and $\rr(t = 0) \equiv \rr$. Let us set
$t' = 0$, obtaining
$$
t = -(\nn\.\rr)\,\frac{h\sinh\ze}{h'} \word{and}
\rr(t) = \rr - (\nn\.\rr)\,\frac{\pp\sinh\ze}{h'}\,.
$$
Then in \eqref{eq:run-along} the second equation becomes:
\begin{align*}
\rr' &= \rr - (\nn\.\rr)\pp\,\frac{\sinh\ze}{h'} 
+ (\nn\.\rr) \nn \bigl( (\cosh\ze - 1)(1 - (\nn\.\pp)\sinh\ze/h') 
- h\sinh^2\ze/h' \bigr)
\\
&= \rr - (\nn\.\rr)\pp\,\frac{\sinh\ze}{h'} 
+ (\nn\.\rr)\nn\,\frac{h}{h'}\,
\bigl( \cosh\ze(\cosh\ze - 1) - \sinh^2\ze \bigr)
\\
&= \rr - (\nn\.\rr)\bigl( \pp\sinh\ze + h\nn(\cosh\ze - 1) \bigr)/h',
\end{align*}
reproducing the external coordinate part of the rule
\eqref{eq:hardened-position}, on the nose. In conclusion, the first
three terms are the expected ones for a structureless particle. Such
an expression as above does not relate two coordinatizations of the
same set of events, but the values of the position coordinates at the
simultaneity hyperplanes $t = 0$ and $t' = 0$ of the two frames
related by the boost. It renders the position coordinates' Lorentz
transformation behaviour in a formulation in which time has been
eliminated: that is to say, the transformations are regarded as acting
on the initial conditions (points of the coadjoint orbit) of a
covariant formulation -- see \cite[Sect.~3]{Ganymede} as well as the
discussion in \cite[Ch.~20]{SudarshanM74}.

In spite of the above, the $\rr$-position coordinates are \textit{not}
relativistically covariant, due to the internal variables. When $\hel$
is just a parameter and $\ttt$ is set to~$\zero$, we are in the
fixed-helicity context; and that is still the case. The phenomenon is
not new: for the massive particles with spin, one can find
\textit{both} (global) canonical and covariant position coordinates;
but they do not coincide. This has been known for a good
while~\cite{FatherAHP80}. The limit of the covariant coordinates $\xx$
in the massive case,
$$
\xx := - \frac{\kk}{h} - \frac{\pp \x \sss}{mh} 
= -\frac{\kk}{h} - \frac{\pp \x \ww}{m^2h}\,,
$$
as $m \downto 0$, $\sss \upto \infty$ is singular, at any rate. Local
canonical coordinates always exist, due to Darboux's theorem -- see 
the next subsection~\ref{ssc:as-the-world-turns}. It is an open question
whether covariant coordinates exist in our case; 
experience~\cite{Ganymede} suggests that it would be rewarding to
work with them.

\subsection{Slant and the symplectic structure}
\label{ssc:as-the-world-turns}

Since $\ttt$ takes values on a circle, it is most natural to regard it
as being given by an angle~$\th$, as is done in \cite{BargmannW48} for
a quantum counterpart. Let $\ttt_1(\pp),\,\ttt_2(\pp) \perp \pp$, with
moreover $\ttt_1(\pp) \perp \ttt_2(\pp)$, be chosen spacelike vectors
of length~$\rho$; and let us expand $\ttt$ in terms of these:
$$
\ttt = \al_1\,\ttt_1(\pp) + \al_2\,\ttt_2(\pp).
$$
Then $|\al_1|^2 + |\al_2|^2 = 1$, and we express $\ttt$ through 
$\al_1 = \cos\th$, $\al_2 = \sin\th$. Such an internal variable $\th$
we baptize here the \textit{slant}. The decomposition is clearly
non-unique, so some choices need to be made. Since $\pp$ and $\la$
Poisson-commute, it is natural to ask that $\{\hel,\ttt_1(\pp)\} = 0$.
In this way,
\begin{align*}
& \al_1(\pp/h \x \ttt_1(\pp)) + \al_2(\pp/h \x \ttt_2(\pp))
= \pp/h \x \ttt
\\
&\qquad = - \{\hel,\ttt\} 
= - \{\hel,\cos\th\}\,\ttt_1(\pp) - \{\hel,\sin\th\}\,\ttt_2(\pp)
- \sin\th \{\hel,\ttt_2(\pp)\}.
\end{align*}
This leads us to choose $\ttt_2(\pp) = \pp/h \x \ttt_1(\pp)$,
therefore $\{\hel,\ttt_2(\pp)\} = 0$ and
$\ttt_1(\pp) = -\pp/h \x \ttt_2(\pp)$, so that
$$
\pp/h \x \ttt = \ttt_2(\pp) \cos\th - \ttt_1(\pp) \sin\th
\word{and} 
\{\hel,\cos\th\} = \sin\th, \quad \{\hel,\sin\th\} = -\cos\th,
$$
that is, $\{\hel,e^{i\th}\} = -ie^{i\th}$. In conclusion, $\hel$ and
$\th$ are symplectically conjugate variables.

Now consider the Poisson brackets involving components of $\rr$
and~$\ttt$. Given that $\{\rr,\hel\} = \zero$, the natural choice is
to take $\{r^a, \cos\th\} = \{r^a, \sin\th\} = 0$, so that the spatial
and internal coordinates symplectically decouple completely. This is
consistent: the only check that we have on our choices so far is that
$\{r^a,\ttt\} = -t^a \pp/h^2$, from~\eqref{eq:taking-position}. This
can be satisfied on deciding for
\begin{align*}
\{r^a, \ttt_1(\pp)\} 
&= - (t^a \cos\th/h^2 + (\pp \x \ttt)^a \sin\th/h^3)\,\pp, \quad
\text{implying}
\\
\{r^a, \ttt_2(\pp)\} 
&= - (t^a \sin\th/h^2 - (\pp \x \ttt)^a \cos\th/h^3)\,\pp; \quad
\text{and}
\\
\{r^a, \ttt\} &= -(\cos^2\th + \sin^2\th)\,t^a\pp/h^2.
\end{align*}
It should be remarked that the moving-frame component of $\rr$ along
$\ttt$ alone does not Poisson-commute with $\ttt$:
$$
\{\pp\.\rr, \ttt\} = \zero = \{[\pp,\ttt,\rr], \ttt\}; \quad
\{\ttt\.\rr, \ttt\} = -(\rho^2/h^2)\,\pp.
$$

\medskip

Schwinger points in~\cite{Schwinger70} to the duality of the spinning
massless relativistic problem with that of an electrically charged
particle in the distant field of a stationary magnetic charge. This
was further explored by Bacry~\cite{Bacry88}. Indeed, in view of
Lemma~\ref{lm:taking-position}, the spatial component of the Poisson
brackets for our problem in the $(r^a,p^b)$ coordinates is of the form
$$
\{f,g\}(\rr,\pp) = \frac{\del f}{\del r^a}\,\frac{\del g}{\del p_a}
- \frac{\del f}{\del p_a}\,\frac{\del g}{\del r^a} 
- \eps^{abc} \hel \frac{p_c}{h^3}\,
\frac{\del f}{\del r^a}\,\frac{\del g}{\del r^b};
$$
where in the fixed-helicity case $\hel$ is replaced by a number.%
\footnote{Since the $\rr$ and $\pp$-variables are conjugate, by
Lemma~\ref{lm:taking-position}, in the rest of this section we shall
write $p_b$ rather than~$p^b$ for the components of~$\pp$, as is
customary.}
We have arrived at a Poisson matrix of the form
$$
\begin{pmatrix}
0 & 1 & -\hel p_3/h^3 & 0 & \hel p_2/h^3 & 0 & 0 & 0 \\
-1 & 0 & 0 & 0 & 0 & 0 & 0 & 0 \\
\hel p_3/h^3 & 0 & 0 & 1 & -\hel p_1/h^3 & 0 & 0 & 0 \\
0 & 0 & -1 & 0 & 0 & 0 & 0 & 0 \\
-\hel p_2/h^3 & 0 & \hel p_1/h^3 & 0 & 0 & 1 & 0 & 0 \\
0 & 0 & 0 & 0 & -1 & 0 & 0 & 0 \\
0 & 0 & 0 & 0 & 0 & 0 & 0 & 1 \\
0 & 0 & 0 & 0 & 0 & 0 & -1 & 0 \end{pmatrix};
$$
and the helicity $\hel$ plays a role dual to a magnetic monopole
charge. That is to say, the symplectic form on the WP orbits is given
by
$$
\om = dr^a \w dp_a - \epsii_{ab}^c \,\hel p_c h^{-3} \,dr^a \w dr^b
+ d\hel \w d\th.
$$
Thus we have obtained a Schwinger--Bacry structure, by now well known
\cite{Soloviev16a,Soloviev16b,Soloviev17}. A quaternionic setting for
it was devised by Emch and Jadczyk in~\cite{EmchJ98}, further explored
by Cari\~nena, Marmo and three of us in~\cite{Monachia,Ravel} and by
Soloviev in~\cite{Soloviev16b,Soloviev17}. It would be most
interesting to know whether it is relevant for the Wigner particle.%
\footnote{The simpler ordinary ``magnetic'' Poisson bracket has been
the object of several studies leading up to an (already rather
inexplicit) magnetic Weyl--Moyal product
\cite{MantoiuP04,KarasevO05}.}

Since the ``magnetic terms'' drop out from $\om^{\w 4}$, a Liouville
measure on the orbit is immediately seen to be
$$
\mu := \om^{\w 4} 
\propto dr^1 \w dr^2 \w dr^3 \w dp_3 \w dp_2 \w dp_1 \w d\hel \w d\th.
$$

By general theory, Darboux or canonical coordinates $q^a$ do always
exist \textit{locally}~\cite{AbrahamM87}. To find them in our case,
the systematic method advocated in \cite{KurnyavkoS16} is not
necessary. For the $6 \x 6$ principal submatrix here they are computed
by the standard physical procedure of defining a ``vector potential''
$\VV(\pp)$ such that
$$
\del_b V_a - \del_a V_b = - \epsii_{ab}^c \,\hel p_c h^{-3}
= \epsii_{ab}^c\,\hel\,\del_c(h^{-1}).
$$
Then $q^a := r^a + V^a$ foots the bill. An example is:
$\VV(\pp) = (p_2,-p_1,0)/|\pp|(p_3 + |\pp|)$.%
\footnote{From this (or a similar) formula it should be clear that
under quantization one expects functional-analytic complications in
the ket~space -- of the kind discussed in~\cite{FlatoSF83}.}

\section{Conclusion and outlook}
\label{sec:open-season}

The renewed excitement on the Wigner particles is a refreshing novelty
in what seemed a foreclosed issue: the identification of physical
particles. Although it is true that the massless WPs analyzed in this
paper are not established members of the present zoo of particles, the
fact that they have been relatively little studied leaves the
possibility of important surprises wide open. While the final aim is a
consistent quantum field theory~\cite{Rehren17}, the Kirillov orbit
method we have described here is a prelude to a full-fledged first
quantization of these particles.

The nonzero Poisson bracket among the coordinates
\eqref{eq:taking-position} also makes the configuration space a
``noncommutative geometry''. One further step in that direction would
be to proceed along the lines of deformation quantization, and present
a Groenewold--Moyal star product, like the ones obtained in
\cite{Ganymede}. This could be done with a generalization of the
invertible Wigner transform to this case; it is a nontrivial task. It
is likely that the definition of a product to low orders in~$\hbar$
\textit{\`a~la} Kontsevich may be more tractable. We recall that a
similar problem is found for the quantization of particles in a
magnetic monopole background
\cite{Soloviev16a,Soloviev16b,Soloviev17,EmchJ98,Monachia,Ravel}. A
related issue is the existence of covariant coordinates, and their
eventual role in a quantization procedure.

Finally: with a few exceptions~\cite{BekaertNS16,KhabarovZ17}, so far
most of the studies of WPs have been made for the bosonic case.
Relatively little is known about them in the spin case, whose
quantization may be of interest.

\appendix

\section{The Poincar\'e Lie algebra}
\label{app:P-Lie-structure}

We briefly summarize our notational conventions for the generators of
the Poincar\'e group.

We use the Minkowski metric whose inner product of $4$-vectors
$x = (x^0,\xx)$, $y = (y^0,\yy)$ is denoted by parentheses:
$(xy) = x^\mu y_\mu := x^0y^0 - \xx\.\yy$. As usual, we write 
$x^2 = (xx) = x^\mu x_\mu$.

The (restricted) Poincar\'e group $\sP$ is the semidirect
product $T_4 \rtimes \sL^\up_+$, with multiplication written as
$(a,\La) \cdot (a',\La') = (a + \La a', \La\La')$. Its Lie 
algebra~$\pl$ has a basis $\{H,P^a,L^a,K^a : a = 1,2,3\}$,
corresponding respectively to time translations, space translations,
rotations and boosts. The nonzero commutation relations are as
follows:
\begin{alignat}{3}
[L^a,L^b] &= \epsi^{ab}_c\, L^c,  \qquad 
& [L^a,K^b] &= \epsi^{ab}_c\, K^c,  \qquad
& [K^a,K^b] &= -\epsi^{ab}_c\, L^c,
\notag \\
[L^a, P^b] &= \epsi^{ab}_c\, P^c,  \qquad 
& [P^b, K^a] &= \dl^{ab} H, \qquad
& [H, K^a] &= P^a.
\label{eq:comm-relns} 
\end{alignat}
The Lorentz-subgroup generators are also denoted by
$J^{0a} := K^a$, $J^{ab} = \epsi^{ab}_c L^c$. Note as well that
$J_{0a} = - J^{0a} = - K^a$, $J_{ab} = J^{ab} = \eps_{abc} L^c$,
and both $J^{\mu\nu}$ and $J_{\rho\sg}$ are skewsymmetric tensors.
The commutation relation of the latter may be summarized as 
$$
[J_{\rho\sg}, J_{\mu\nu}] 
= -g_{\rho\mu} J_{\sg\nu} - g_{\sg\nu} J_{\rho\mu} 
+ g_{\sg\mu} J_{\rho\nu} + g_{\rho\nu} J_{\sg\mu} \,.
$$

The dual tensor 
$J^{*\rho\mu} := -\half \eps^{\rho\mu\nu\tau} J_{\nu\tau}$ plays a
role in the theory of the WP. Here $J^{*0a} = - L^a$ and
$J^{*ab} = \epsi^{ab}_c K^c$. Remark that 
$$
\KK \cdot\LL = \half J_{\rho\mu} J^{*\rho\mu}  \word{and}
\KK^2 - \LL^2 = \half J_{\rho\mu} J^{\rho\mu} 
= -\half J^*_{\rho\mu} J^{*\rho\mu}
$$
are the Casimirs of the Lorentz group.

\section{A pedestrian approach to the ``Pauli--Luba\'nski limit''}
\label{app:PL-Schwinger}

In our classical context, it is possible to obtain the WP data by
carefully taking the $m \downto 0$, $\sss \upto \infty$ limit (with
$m\sss$ finite). This we do essentially following Schwinger,
\textit{mutatis mutandis}. To the purpose, let us go back to equation
\eqref{eq:mass-spin}, where it is clear that
$$
\lim_{m\downto 0,\,\sss\upto\infty} m\sss = \ttt \perp \pp.
$$
Clearly as well, from $\{s^a,s^b\} = \epsi^{ab}_c s^c$ we infer
$\{t^a,t^b\} = 0$ and $\{\hel,\ttt\} = -\pp/h \x \ttt$, as well as the
other brackets in Lemma~\ref{lm:tea-table}. The interesting part is
that instead of $\qq$ as given in~\eqref{eq:take-a-position}, which is
ill-defined in the $m \downto 0$, $\sss \upto \infty$ limit, one can
define already an $\rr$-vector in the massive case enjoying a smooth
Pauli--Luba\'nski limit by
$$
\rr := \qq + \frac{\pp \x \sss}{h^2} 
= - \frac{\kk}{h} - \frac{\pp \x \sss}{h(m + h)}
+ \frac{\pp \x \sss}{h^2} 
= - \frac{\kk}{h} + \frac{\pp \x m\sss}{h^2(m + h)}\,;
$$
and already \eqref{eq:keep-apart} holds, too, as well as the
``magnetic'' commutation relations between the $r^a$ given
in~\eqref{eq:taking-position}:
\begin{align*}
\{r^a,r^b\} &= \{q^a, \epsi^b_{ef}\, p^e s^f h^{-2}\}
+ \{\epsi^a_{cd}\, p^c s^d h^{-2}, q^b\} 
+ \epsi^a_{cd}\,\epsi^b_{ef}\, p^c p^e \{s^d, s^f\} h^{-4}
\\
&= -2 \epsi^{ab}_c s^c h^{-2} - 2 \epsi^b_{cd}\, p^a p^c s^d h^{-4}
+ 2 \epsi^a_{cd}\, p^b p^c s^d h^{-4} 
+ \epsi^a_{cd} \epsi^b_{ef} \epsi^{df}_g\, p^c p^e s^g h^{-4}
\\
&= -2 \epsi^{ab}_c s^c h^{-2} - 2 p^a (\pp \x \sss)^b h^{-4}
+ 2 p^b (\pp \x \sss)^a h^{-4} + \epsi^{ab}_c \, p^c w^0 h^{-4},
\end{align*}
on using $w^0 = \pp\.\sss$ in the massive case. Since 
$\pp \x (\pp \x \sss) = w^0 \pp - h^2 \sss$, this reduces to
$$
\{r^a,r^b\} = \epsi^{ab}_c \bigl( -2 s^c h^{-2}
- 2(w^0 p^c - h^2 s^c) h^{-4} + w^0 p^c h^{-4} \bigr)
= - \epsi^{ab}_c\, w^0 p^c h^{-4}
= - \epsi^{ab}_c\, \hel p^c h^{-3},
$$
in agreement with the last formula of Lemma~\ref{lm:taking-position}.

\subsection*{Acknowledgements}
We are grateful for discussions with Alejandro Jenkins -- and his
astute comments regarding the gyroscope. We thank the referee for a
careful reading and for suggesting welcome clarifications. JMG-B and
JCV heartily thank the Universit\`a di Napoli \textit{Federico~II} and
INFN Sezione di Napoli for its hospitality, and JMG-B wishes also to
thank the ZiF at the University of Bielefeld for its warm welcome at
the final stage of the paper. JCV thanks the Instituto de F\'isica
Te\'orica UAM--CSIC for providing a splendid working environment at
Madrid, during the late stages of writing.

The project has received funding from the European Union's
Horizon~2020 research and innovation programme under the Marie
Sk{\l}odowska-Curie grant agreement No.~690575. JMG-B, FL and PV
acknowledge the support of the COST action QSPACE; FL and PV
acknowledge as well the Iniziativa Specifica GeoSymQFT of the INFN. FL
acknowledges the Spanish MINECO under project MDM--2014--0369 of ICCUB
(Unidad de Excelencia `Mar\'ia de Maeztu'). JMG-B received funding
from Project FPA2015--65745--P of MINECO/Feder. JCV acknowledges
support from the Vicerrector\'ia de Investigaci\'on of the Universidad
de Costa~Rica.


\begin{thebibliography}{37}

\footnotesize

\bibitem{Wigner39}
E. P. Wigner,
``On unitary representations of the inhomogeneous Lorentz group'',
Ann. Math. \textbf{40} (1939), 149--204.

\bibitem{Kirillov04}
A. A. Kirillov,
\textit{Lectures on the Orbit Method},
American Mathematical Society, Providence, RI,~2004.

\bibitem{Ganymede}
J. F. Cari\~nena, J. M. Gracia--Bond\'ia and J. C. V\'arilly,
``Relativistic quantum kinematics in the Moyal representation'',
J. Phys. A\,\textbf{23} (1990), 901--933.

\bibitem{Li93}
Z. Li,
``Coadjoint orbits and induced representations'',
Ph.\,D. thesis, MIT, 1993.

\bibitem{SToro13a}
Ph. Schuster and N. Toro,
``On the theory of continuous spin particles: wavefunctions and
soft-factor scattering amplitudes'',
JHEP \textbf{1309} (2013), 104.

\bibitem{SToro13b}
Ph. Schuster and N. Toro,
``On the theory of continuous-spin particles: helicity correspondence
in radiation and forces'',
JHEP \textbf{1309} (2013), 105.

\bibitem{SToro13c}
Ph. Schuster and N. Toro,
``A gauge field theory of continuous spin particles'',
JHEP \textbf{1310} (2013), 061.

\bibitem{SToro15}
Ph. Schuster and N. Toro,
``Continuous-spin particle field theory with helicity 
correspondence'',
Phys. Rev. D \textbf{91} (2015), 025023.

\bibitem{RempelF16}
T. Rempel and L. Freidel,
``A bilocal model for the relativistic spinning particle'',
Phys. Rev. D \textbf{95} (2017), 104014.

\bibitem{BekaertS17}
X. Bekaert and E. Skvortsov,
``Elementary particles with continuous spin'',
Int. J. Mod. Phys. A \textbf{32} (2017), 1730019.

\bibitem{Najafizadeh17}
M. Najafizadeh,
``Modified Wigner equations and continuous spin gauge field'',
Phys. Rev. D \textbf{97} (2018), 065009.

\bibitem{Rehren17}
K.-H. Rehren,
``Pauli--Luba\'nski limit and stress-energy tensor for infinite-spin
fields'',
JHEP \textbf{1711} (2017), 130.

\bibitem{DufloV69}
M. Duflo and M. Vergne,
``Une propri\'et\'e de la representation coadjointe d'une alg\`ebre 
de~Lie'',
C. R. Acad. Sci. Paris \textbf{268A} (1969), 583--585.

\bibitem{Wightman62}
A. S. Wightman,
``On the localizability of quantum mechanical systems'',
Rev. Mod. Phys. \textbf{34} (1962), 845--872.

\bibitem{Schwinger70}
J. Schwinger,
\textit{Particles, Sources and Fields},
Addison-Wesley, Reading, MA, 1970.

\bibitem{BalachandranMSSZ92}
A. P. Balachandran, G. Marmo, A. Simoni, A. Stern and F. Zaccaria,
``On a classical description of massless particles'',
Proceedings of the ISAQTP--Shanxi (1992); pp.~396--402.

\bibitem{BoyaCS74}
L. J. Boya, J. F. Cari\~nena and M. Santander,
``On the continuity of boosts for each orbit'',
Commun. Math. Phys. \textbf{37} (1974), 331--334.

\bibitem{Wigner47}
E. P. Wigner,
``Relativistische Wellengleichungen'',
Z. Physik \textbf{124} (1947), 665--684.

\bibitem{Eunike}
J. M. Gracia-Bond\'ia, A. Jenkins and J. C. V\'arilly,
in preparation.

\bibitem{AliKM00}
S. T. Ali, A. E. Krasowska and R. Murenzi,
``Wigner functions from the two-dimensional wavelet group'',
J. Opt. Soc. Am. A\,\textbf{17} (2000), 2277--2287.

\bibitem{SudarshanM74}
E. C. G. Sudarshan and N. Mukunda,
\textit{Classical Dynamics: A Modern Perspective},
Wiley, New York, 1974.

\bibitem{FatherAHP80}
L. Bel and J. Mart\'in,
``Predictive relativistic mechanics of $N$ particles with spin'',
Ann. Inst. Henri Poincar\'e A\,\textbf{33} (1980), 409--442.

\bibitem{BargmannW48}
V. Bargmann and E. P. Wigner,
``Group theoretical discussion of relativistic wave equations'',
Proc. Natl. Acad. Sci. USA \textbf{34} (1948), 211--223.

\bibitem{Bacry88}
H. Bacry,
\textit{Localizability and Space in Quantum Physics},
Springer, Berlin, 1988.

\bibitem{Soloviev16a}
M. A. Soloviev,
``Weyl correspondence for a charged particle in the field of a 
magnetic monopole'',
Theor. Math. Phys. \textbf{187} (2016), 782--795.

\bibitem{Soloviev16b}
M. A. Soloviev,
``Dirac's monopole, quaternions, and the Zassenhaus formula'',
Phys. Rev. D\,\textbf{94} (2016), 105021.

\bibitem{Soloviev17}
M. A. Soloviev,
``Dirac's magnetic monopole and the Kontsevich star product'',
J. Phys. A: Math. Theor. \textbf{51} (2018), 095205.

\bibitem{EmchJ98}
G. G. Emch and A. Z. Jadczyk,
``Weakly projective representations, quaternions and monopoles'',
in \textit{Stochastic Processes, Physics and Geometry: 
New Interplays, II},
F. Gesztesy \textit{et~al}, eds.,
CMS Conference Proceedings \textbf{29}, 
AMS, Providence, RI, 2000; pp.~157--164.

\bibitem{Monachia}
J. F. Cari\~nena, J. M. Gracia-Bond\'ia, F. Lizzi, G.~Marmo and
P.~Vitale,
``Monopole-based quantization: a programme'',
in: \textit{Mathematical Physics and Field Theory: Julio Abad, in
memoriam}, M. Asorey, J. V. Garc\'ia-Esteve, M. Fern\'andez-Ra\~nada
and J. Sesma, eds., Prensas Universitarias de Zaragoza, 2009;
pp.~167--176.

\bibitem{Ravel}
J. F. Cari\~nena, J. M. Gracia-Bond\'ia, F. Lizzi, G.~Marmo and
P.~Vitale,
``Star-product in the presence of a monopole'',
Phys. Lett. A\,\textbf{374} (2010), 3614--3618.

\bibitem{MantoiuP04}
M. M\u{a}ntoiu and R. Purice,
``The magnetic Weyl calculus'',
J. Math. Phys. \textbf{45} (2004), 1394--1417.

\bibitem{KarasevO05}
M. V. Karasev and T. A. Osborn,
``Cotangent bundle quantization: entangling of metric and magnetic
field'',
J.~Phys. A: Math. Gen. \textbf{38} (2005), 8549--8578.

\bibitem{AbrahamM87}
R. Abraham and J. E. Marsden,
\textit{Foundations of Mechanics},
Addison-Wesley, Redwood City, CA, 1987.

\bibitem{KurnyavkoS16}
O. L. Kurnyavko and I. V. Shirokov,
``Construction of invariants of the coadjoint representation of Lie
groups using linear algebra methods'',
Theor. Math. Phys. \textbf{188} (2016), 965--979.

\bibitem{FlatoSF83}
M. Flato, D. Sternheimer and C. Fr{\o}nsdal,
``Difficulties with massless particles?''
Commun. Math. Phys. \textbf{90} (1983), 563--573.

\bibitem{BekaertNS16}
X. Bekaert, M. Najafizadeh and M. R. Setare,
``A gauge field theory of fermionic continuous-spin particles'',
Phys. Lett. B \textbf{760} (2016), 320--323.

\bibitem{KhabarovZ17}
M. V. Khabarov and Yu. M. Zinoviev,
``Infinite (continuous) spin fields in the frame-like formalism'',
Nucl. Phys. B \textbf{928} (2018), 182--216.


\end{thebibliography}
\end{document}